\documentclass[final]{IEEEtran}

\usepackage{ifpdf}
\ifpdf
	\usepackage[pdftex]{graphicx}
  \usepackage[update]{epstopdf}
\else
	\usepackage{graphicx}
\fi
\usepackage{amsmath,amssymb,array,stfloats,cuted,midfloat,psfrag,setspace}
\usepackage[noadjust]{cite}
\usepackage{hyperref}

\newtheorem{theorem}{\textbf{Theorem}}
\newtheorem{proposition}{\textbf{Proposition}}

\newtheorem{definition}{\textbf{Definition}}

\newtheorem{corollary}{\textbf{Corollary}}
\newtheorem{lemma}{\textbf{Lemma}}

\newcommand{\Expt}{\mbox{${\mathbb E}$} }

\begin{document}
\date{}

\title{Approximate Characterizations for the Gaussian Source Broadcast Distortion Region}

\author{Chao Tian,~\IEEEmembership{Member,~IEEE}, Suhas Diggavi,~\IEEEmembership{Member,~IEEE}, \\and Shlomo Shamai (Shitz),~\IEEEmembership{Fellow,~IEEE}
\thanks{The material in this paper was presented in part at the IEEE International Symposium on Information Theory, Seoul, Korea, June-July 2009.}
\thanks{The work of S. Shamai was supported by the European Commission in the framework of the FP7 Network of Excellence in Wireless COMmunications, NEWCOM++.}
\thanks{C. Tian is with AT\&T Labs-Research, Florham Park, NJ 07932, USA. (email: tian@research.att.com)}
\thanks{S. N. Diggavi is with the Department of Electrical Engineering, University of
California, Los Angeles, CA 90095, USA. (email: suhas@ee.ucla.edu)}
\thanks{S. Shamai is with the Department of Electrical Engineering, Technion--Israel Institute of Technology, Haifa 32000, Israel. (email: sshlomo@ee.technion.ac.il)}
}

\maketitle

\begin{abstract}
We consider the joint source-channel coding problem of sending a Gaussian source on a $K$-user Gaussian broadcast channel with bandwidth mismatch. A new outer bound to the achievable distortion region is derived using the technique of introducing more than one additional auxiliary random variable, which was previously used to derive sum-rate lower bound for the symmetric Gaussian multiple description problem. By combining this outer bound with the achievability result based on source-channel separation, we provide approximate characterizations of the achievable distortion region within constant multiplicative factors. Furthermore, we show that the results can be extended to general broadcast channels, and the performance of the source-channel separation based approach is also within the same constant multiplicative factors of the optimum. 
\end{abstract}
\begin{keywords}
Gaussian source, joint source-channel coding, squared error distortion.
\end{keywords}

\section{Introduction}
\label{sec:intro}

Shannon's source-channel separation theorem essentially states that asymptotically there is no loss from optimum by decoupling the source coding component and channel coding component in a point-to-point communication system \cite{Shannon:48}. This separation result tremendously simplifies the concept and design of communication systems, and it is also the main reason for the division between research in source coding and channel coding. However, it is also well known that in many multi-user settings, such a separation indeed incurs certain performance loss; see, \textit{e.g.}, \cite{Goblick:65,Cover:80,Gastpar:03}. For this reason, joint source-channel coding has attracted an increasing amount of attention as the communication systems become more and more complex.

One of the most intriguing problems in this area is joint source-channel coding of a Gaussian source on a Gaussian broadcast channel with $K$ users under an average power constraint. It was observed by Goblick \cite{Goblick:65}  that when the source bandwidth and the channel bandwidth are matched, \textit{i.e.}, one channel use per source sample, directly sending the source samples on the channel after a simple scaling is in fact optimal, but the separation-based scheme suffers a performance loss \cite{Gastpar:03}. However, when the source bandwidth and the channel bandwidth are not matched, such a simple scheme is no longer optimal. Many researchers have considered this problem, and significant progress has been made toward finding better coding schemes based on hybrid digital and analog signaling; see, \textit{e.g.}, \cite{ShamaiVerduZamir:98,Mittal:02,Skoglund:06,ReznicFederZamir:06,Prabhakaran:05,NarayananCaireWilson:07} and the references therein.

In spite of the progress on the achievability schemes, our overall understanding on this problem is still quite limited. As pointed out by Caire \cite{Caire:ITA06}, the key difficulty appears to be finding meaningful outer bounds. Such outer bounds not only can provide a concrete basis to evaluate various achievability schemes, but also may provide insights into the structure of good or even optimal codes, and may further suggest simplification of the possibly quite complex optimal schemes in certain distortion regimes. In this regard, the result by Reznic \textit{et al.} \cite{ReznicFederZamir:06} is particularly important, where they derived a non-trivial outer bound for the achievable distortion region for the two-user system. This outer bound relies on a technique previously used in the multiple description problem by Ozarow \cite{Ozarow:80}, where one additional random variable beyond those in the original problem is introduced. The bound given in \cite{ReznicFederZamir:06} is however rather complicated, and was only shown to be asymptotically tight for certain high signal to noise ratio regime. 

In this work, we derive an outer bound for the $K$-user problem using a similar technique as that used in \cite{ReznicFederZamir:06}, however, more than one additional random variable is introduced. The technique used here also bears some similarity to that used in \cite{Tian:08}. The outer bound has a more concise form than the one given in \cite{ReznicFederZamir:06}, but for the $K=2$ case, it can be shown that they are equivalent. This outer bound is in fact a set of outer bounds parametrized by $K-1$ non-negative variables. Though one can optimize over these variables to find the tightest one, this optimization problem appears difficult. Thus we take an approach similar to the one taken in \cite{Tian:08}, and choose some specific values for the $K-1$ variables which gives specific outer bounds. Moreover, by combining these specific outer bounds with the simple achievability scheme based on source-channel separation, we provide approximate characterizations\footnote{We would like to thank David Tse for discussions in ITA 2008 on the formulation of the question, where he was motivated by his solution to a deterministic version of this problem.} of the achievable distortion region within some universal constant multiplicative factors, independent of the signal to noise ratio and the bandwidth mismatch factor. In one of the approximations, the multiplicative factor is roughly of form $2^k$ for the distortion at the $k$-th user, while in the other, the factor is $K$ for all the distortions. Thus although Shannon's source-channel separation result does not hold strictly in this problem, it indeed holds in an approximate manner. In fact, this set of results is extremely flexible, and can be applied in the case with an infinite number of users but the minimum achievable distortion is bounded away from zero, for which we can conclude that the source-channel separation based approach is also within certain finite constant multiplicative factors of the optimum. In this case, these constants can be upper bounded by factors related to the disparity between the best and worse distortions, which is not influenced by the number of users being infinite. 

Though the outer bound is derived using techniques that have some precedents in the information theory literature, the difficulty lies in determining which terms to bound. In contrast to pure source coding problems or pure channel coding problems, where we can usually meaningfully bound a linear combination of rates, in a joint source-channel coding problem the notion of rates does not exist. In  \cite{ReznicFederZamir:06}, the lower bound on one distortion is given in terms of a function of the other distortion in the two-user problem. It is clear that such a proof approach becomes unwieldy for the general $K$-user case. In this work, we instead derive bounds for some quantity which at the first sight may seem even unrelated to the problem, but eventually serves as an interface between the source and channel coding components, thus replacing the role of ``rates'' in traditional Shannon theory proofs.

Inspired by a recent work of Avestimehr, Caire and Tse \cite{Avestimehr:08}, where source-channel separation in more general networks is considered, we further show that our technique can be conveniently extended to general broadcast channels, and the source-channel separation based scheme is within the same multiplicative constants of the optimum as for the Gaussian channel case. 

The rest of the paper is organized as follows. Section \ref{sec:definition} gives the necessary notation and reviews an important lemma useful in deriving the outer bound. The main results are presented in Section \ref{sec:mainresult}, and the proofs for these results are given in Section \ref{sec:proof}. The extension to general broadcast channels is given in Section \ref{sec:general}, and Section \ref{sec:conclusion} concludes the paper. 

\section{Problem Definition and Review}
\label{sec:definition}

In this section, we give a formal definition of the Gaussian source broadcast problem in the context of Gaussian broadcast channels; the notation will be generalized in Section \ref{sec:general} when other broadcast channels are considered. 

Let $\{S(i)\}_{i=1,2,...}$ be a stationary and memoryless Gaussian source with zero-mean and unit-variance. The
vector $(S(1),S(2),...,S(m))$ will be denoted as $S^m$. We use $\mathbb{R}$ to denote the domain of reals, and $\mathbb{R}_+$ to denote the domain of non-negative reals. The Gaussian memoryless broadcast channel is given by the model 
\begin{align}
Y_k=X+Z_k,\quad k =1,2,\ldots,K,
\end{align}
where $Y_k$ is the channel output observed by the $k$-th receiver, and $Z_k$ is the zero-mean additive Gaussian noise on the channel input $X$. Thus the channel is memoryless in the sense that $(Z_1(i),Z_2(i),\ldots,Z_K(i))_{i=1,2,\ldots}$ is a stationary and memoryless process. The variance of $Z_k$ is denoted as $N_k$, and without loss of generality, we shall assume 
\begin{align}
N_1\geq N_2\geq \ldots\geq N_K.
\end{align}
The mean squared error distortion measure is used, which is given by  
$d(s^m,\hat{s}^m)=\frac{1}{m}\sum_{i=1}^m(s(i)-\hat{s}(i))^2$. The encoder maps a source sample block of length $m$ into a channel input block of length $n$, and each decoder maps the corresponding channel output block of length $n$ into a source reconstruction block of length $m$. The bandwidth mismatch factor is thus defined as
\begin{align}
b=\frac{n}{m},
\end{align}
which is essentially the (possibly fractional) channel uses per source sample; see Fig. \ref{fig:systemdiag}. The channel input is subject to an average power constraint.

\begin{figure}[tb]
\begin{centering}
\includegraphics[width=8cm]{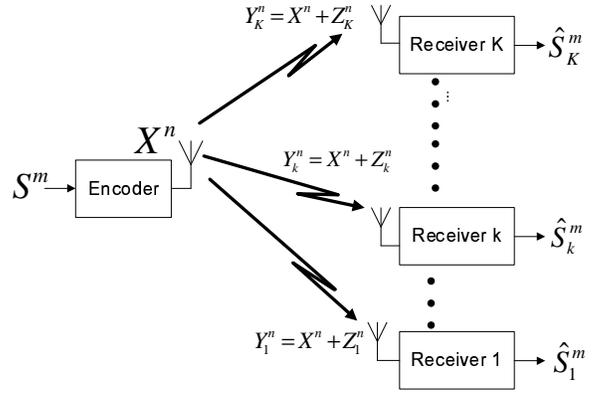}
\caption{Joint source-channel coding with bandwidth mismatch.\label{fig:systemdiag}}
\end{centering}
\end{figure}

We can make the codes in consideration more precise by introducing the following definition.
\begin{definition}
An $(m,n,P,d_1,d_2,\ldots,d_K)$ Gaussian source-channel broadcast code is given by an encoding function
\begin{align}
f:\mathbb{R}^m\rightarrow \mathbb{R}^n,
\end{align}
such that
\begin{align}
\frac{1}{n}\sum_{i=1}^n\Expt (X(i))^2\leq P,
\end{align}
and $K$ decoding functions
\begin{align}
g_k:\mathbb{R}^n\rightarrow \mathbb{R}^m,\quad k =1,2,\ldots,K,
\end{align}
and their induced distortions
\begin{align}
d_k=\Expt d(S^m,g_k(f(S^m)+Z^n_k)),\quad k=1,2,\ldots,K,
\label{eqn:expectation}
\end{align}
where $\Expt(\cdot)$ is the expectation operation.
\end{definition}

Note that there are two kinds of independent randomness in the system, the first of which is by the source, and the second is by the channel noises; the expectation operation in (\ref{eqn:expectation}) is taken over both of them. In the definition, $+$ in the expression $f(S^m)+Z^n_k$  is understood as the length-$n$ vector addition.

From the above definition, it is clear that the performance of any Gaussian joint source-channel code depends only on the marginal distribution of $(S^m,X^n,Y^n_k)$, but not the joint distribution $(S^m,X^n,Y^n_1,Y^n_2,\ldots,Y^n_K)$. This implies that physical degradedness does not differ from statistical degradedness in terms of the system performance. Since the Gaussian broadcast channel is always statistically degraded, we shall assume physical degradedness from here on without loss of generality. The channel noises can thus be written as
\begin{align}
Z_k=Z_{k+1}+\Delta Z_{k},\quad k=1,2,3,\ldots, K,
\end{align}
where $\Delta Z_k$ is a zero-mean Gaussian random variable with variance $\Delta N_k=N_k-N_{k+1}$, which is independent of everything else; for convenience, we define $Z_{K+1}\triangleq 0$, and it follows that $\Delta N_K=N_K$ and $Y_{K+1}=X$.

\begin{definition}
\label{def:distortionvector}
A distortion vector $(D_1,D_2,\ldots,D_K)\in \mathbb{R}_+^K$, where $1\geq D_1\geq D_2\geq\ldots\geq 0$ is achievable under power constraint $P$ and bandwidth mismatch factor $b$, if for any $\epsilon>0$ and sufficiently large $m$, there exist an integer $n\leq bm$ and an $(m,n,P,d_1,d_2,\ldots,d_K)$ Gaussian source-channel broadcast code such that
\begin{align}
D_i+\epsilon\geq d_i,\quad i=1,2,\ldots,K.
\end{align}
\end{definition}

Note that the constraint $1\geq D_1\geq D_2\geq\ldots\geq 0$ is without loss of generality, because otherwise the problem can be reduced to an alternative one with fewer users due to the assumed physical degradedness. The collection of all the achievable distortion vectors under power constraint $P$ and bandwidth mismatch factor $b$ is denoted by $\mathcal{D}(P,b)$, and this is the region that we are interested in.

One important result we need in this work is the following lemma, which is a slightly different version of the one given in \cite{Tian:08}. 
\begin{lemma}
\label{lemma:difference}
Let $W$ be a random variable jointly distributed with the Gaussian source vector $S^m$ in the alphabet $\mathcal{W}$, such that there exists a deterministic mapping $g:\mathcal{W}\rightarrow \mathbb{R}^m$ satisfying
\begin{align}
\Expt d(S^m,g(W))\leq D.
\end{align}
Let $U=S+V$ and $U'=S+V+V'$, where $V$ and $V'$ are mutually independent Gaussian random variables independent of the Gaussian source $S$ and the random variable $W$, with variance $\sigma^2$ and $\sigma'^2$, respectively. Then with $\sigma^2\triangleq\tau$ and $\sigma^2+\sigma'^2\triangleq\tau'$, we have
\begin{enumerate}
\item \textbf{Mutual information bound}
\begin{eqnarray}
\label{eqn:lemmabound1}
I(W;U'^m)\geq \frac{m}{2}\log\frac{1+\tau'}{D+\tau'},
\end{eqnarray}
\item \textbf{Bound on mutual information difference}
\begin{align}
&I(W;U^m)-I(W;U'^m)\geq  \frac{m}{2}\log\frac{(1+\tau)(D+\tau')}{(1+\tau')(D+\tau)}.\label{eqn:lemmabound2}
\end{align}
\end{enumerate}
\end{lemma}

The proof of this lemma is almost identical to the one given in \cite{Tian:08}. The only difference between the two versions is that in \cite{Tian:08} the random variable $W$ is in fact a deterministic function of $S^m$, however it is rather straightforward to verify that this condition was never used in the proof given in \cite{Tian:08}; we include the proof of this lemma in the Appendix for completeness. 

\section{Main Results for Gaussian Broadcast Channels}
\label{sec:mainresult}

Our main results for Gaussian source broadcast on Gaussian broadcast channels are summarized in Theorem \ref{theorem:maintheorem}, Corollary \ref{corollary:firstcorollary},  Proposition \ref{proposition:firstcorollary}, Corollary \ref{corollary:infiniteusers} and Corollary \ref{corollary:secondcorollary}, the proofs of which are given in the next section; extensions of these results to general broadcast channels are given in Section \ref{sec:general}.

Define the region in (\ref{eqn:innerbound}) on the top of next page, which is in fact the inner bound via source-channel separation. Next define the regions in (\ref{eqn:firstoutbound}) and (\ref{eqn:secondouterbound}) also on the top of next page, which are in fact outer bounds to the achievable distortion region. We have the the following theorem.

\begin{figure*}[tb]
\normalsize
\newcounter{MYtempeqncnt}
\setcounter{MYtempeqncnt}{\value{equation}}
\addtocounter{MYtempeqncnt}{1}
\begin{align}
&\hat{\mathcal{D}}(P,b)\triangleq\left\{(D_1,D_2,\ldots,D_K):\sum_{k=1}^K \Delta N_k D_k^{-\frac{1}{b}}\leq P+N_1,\quad 1\geq D_1\geq D_2 \geq\ldots\geq D_K\geq 0 \right\}.
\label{eqn:innerbound}\\
&\underline{\mathcal{D}}^*(P,b)\triangleq\left\{(D_1,D_2,\ldots,D_K):\sum_{k=1}^K \Delta N_k (2^kD_k)^{-\frac{1}{b}}\leq P+N_1,\quad 1\geq D_1\geq D_2 \geq\ldots\geq D_K\geq 0\right\}.\label{eqn:firstoutbound}\\
&\underline{\mathcal{D}}(P,b)\triangleq\left\{(D_1,D_2,\ldots,D_K):\sum_{k=1}^K \Delta N_k (KD_k)^{-\frac{1}{b}}\leq P+N_1,\quad 1\geq D_1\geq D_2 \geq\ldots\geq D_K\geq 0\right\}.\label{eqn:secondouterbound}
\end{align}
\hrulefill
\end{figure*}
\setcounter{equation}{\value{MYtempeqncnt}}
\addtocounter{equation}{2}

\begin{theorem}\label{theorem:maintheorem}
\begin{align}
\hat{\mathcal{D}}(P,b)\subseteq \mathcal{D}(P,b)\subseteq  \underline{\mathcal{D}}^*(P,b)\cap \underline{\mathcal{D}}(P,b).
\end{align}
\end{theorem}


Theorem \ref{theorem:maintheorem} is stated as inner and outer bounds to the achievable distortion region, however it can be observed that the bounds have similar forms, and their difference, in terms of distortions, can be bounded by some multiplicative constants. The following corollary follows directly from Theorem \ref{theorem:maintheorem}, by comparing (\ref{eqn:innerbound}) and (\ref{eqn:firstoutbound}).

\begin{corollary}
\label{corollary:firstcorollary}
If $(D_1,D_2,\ldots,D_K)\in \mathcal{D}(P,b)$, and if $D_k\geq 2D_{k+1}$ for $k=1,2,\ldots,K-1$, then $(2D_1,2^2D_2,\ldots,2^KD_K)\in \hat{\mathcal{D}}(P,b)$.
\end{corollary}

The condition $D_k\geq 2D_{k+1}$ in Corollary \ref{corollary:firstcorollary} is to ensure that the distortion vector $(2D_1,2^2D_2,\ldots,2^KD_K)$ satisfies the monotonicity requirement in Definition \ref{def:distortionvector} and (\ref{eqn:innerbound}). This result has the following intuitive interpretation if the condition indeed holds that $D_k\geq 2 D_{k+1}$ for all $k=1,2,\ldots,K-1$: if a genie helps the separation-based scheme by giving each individual user half a bit information per source sample, and at the same time all the better users also receive this half a bit information for free, then the separation-based scheme is as good as the optimal scheme.

This approximation can in fact be refined, and for this purpose, the following additional definition is needed. For any $1\geq D_1\geq D_2\geq \ldots\geq D_K\geq 0$, we associate with it a \textit{relaxed distortion vector} $(D^*_1,D^*_2,\ldots,D^*_K)$ and a binary labeling vector $(B_1,B_2,\ldots,B_K)$ in a recursive manner
\begin{align}
\label{eqn:enhanceddistortion}
(D^*_k,B_k)&=\left\{\begin{array}{ll}(D^*_{k-1},0)& \mbox{if }2^{1+\sum_{j=1}^{k-1}B_j}\frac{D_k}{D^*_{k-1}}\geq 1 \\(2^{1+\sum_{j=1}^{k-1}B_j}D_k,1)&\mbox{otherwise}\end{array}\right.
\end{align}
for $k=1,2,\ldots,K$, and we have defined $D^*_0=1$ for convenience. It is easily verified that $D^*_k\geq D^*_{k+1}$ for $k=1,2,\ldots,K-1$, and moreover $D^*_k\leq 2^{1+\sum_{j=1}^{k-1}B_j} D_k$.

\begin{proposition}
\label{proposition:firstcorollary}
Let $(D^*_1,D^*_2,\ldots,D^*_K)$ be the relaxed distortion vector of $(D_1,D_2,\ldots,D_K)$. If $(D_1,D_2,\ldots,D_K)\in \mathcal{D}(P,b)$, then $(D^*_1,D^*_2,\ldots,D^*_K)\in \hat{\mathcal{D}}(P,b)$.
\end{proposition}

The notion of relaxed distortion vector essentially removes the rather artificial condition $D_k\geq 2D_{k+1}$ in Corollary \ref{corollary:firstcorollary}. When this condition does not hold for some $k$, the relaxed distortion vector is introduced to replace $(2D_1,2^2D_2,\ldots,2^KD_K)$, which in this case does not satisfy the monotonicity requirement in Definition \ref{def:distortionvector} and thus is not a valid choice of a distortion vector; nevertheless, in this case, the difference between the original distortion vector and its relaxed version is in fact smaller, being $2^{1+\sum_{j=1}^{k-1}B_j}$, instead of $2^k$ for $D_k$ as in the case already considered in Corollary \ref{corollary:firstcorollary}. 

Proposition \ref{proposition:firstcorollary} can be used in the situation where there are an infinite number of users such as in a fading channel. Let the set of users indexed by $x$ and their associated distortions be denoted as $D_x$, since there may be an uncountably infinite many of them. If we apply the construction given in (\ref{eqn:enhanceddistortion}), with $B_i$ replaced by $B_x$, $\sup\{D_x\}$ taking the role of $D_1$ and $\inf\{D_x\}$ taking the role of $D_K$, then the following lemma is straightforward by observing that $\inf\{D^*_x\}\leq \sup\{D^*_x\}\leq 2\sup\{D_x\}$ and $\inf\{D^*_x\}=2^{\sum_x{B_x}} \inf\{D_x\}$.
\begin{lemma}
\label{lemma:maximumgroups}
The sequence $B_x$ specified by (\ref{eqn:enhanceddistortion}) satisfies $\sum_{x}B_x\leq \log_2(\sup\{D_x\})-\log_2(\inf\{D_x\})+1$.
\end{lemma}

It is clear that the maximum multiplicative constant is less than $2^{1+\sum_{x}B_x}$ in the statement of Proposition \ref{proposition:firstcorollary}. If there exists a lower bound on the achievable distortion for the best user, denoted as $d_{\mbox{\small{min}}}$, which is strictly positive, {\em i.e.}, $\inf\{D_x\}\geq d_{\mbox{\small{min}}}>0$, then since $\sup\{D_x\}\leq 1$, the multiplicative factor can be bounded as 
\begin{align*}
2^{1+\sum_{x}B_x}\leq \frac{4}{d_{\mbox{\small{min}}}}.
\end{align*} 
Thus even when the number of users is infinite, as long as the lower bound $d_{\mbox{\small{min}}}$ is bounded away from zero, the multiplicative factors are in fact finite. More formally, we have the following corollary\footnote{Here we directly take the number of users to infinity in Proposition \ref{proposition:firstcorollary}, however a more rigorous approach is to derive the outer bounds for this case and show the result holds. This can indeed be done either along the line of the proof given in Section \ref{sec:proof} with careful replacement of summation by integral, or more straightforwardly along the line of proof given in Section \ref{sec:general}.}.

\begin{corollary}
\label{corollary:infiniteusers}
For an infinite number of users indexed by $x$ with $\inf_{\{D_x\}\in\mathcal{D}(P,b)}\inf\{D_x\}\geq d_{\mbox{\small{min}}}$, let $\{D^*_x\}$ be the relaxed distortion vector of $\{D_x\}$. If $\{D_x\}\in \mathcal{D}(P,b)$, then $\{D^*_x\}\in \hat{\mathcal{D}}(P,b)$, and furthermore, $\sup\{\frac{D_x}{D^*_x}\}\leq \frac{4}{d_{\mbox{\tiny{min}}}}$. 
\end{corollary}

The next corollary gives another version of the approximation, essentially stating that for any achievable distortion vector, its $K$-fold multiple is achievable using the separation approach. In terms of the genie-aided interpretation, the genie only needs to provide $\frac{1}{2}\log K$ bits common information to the users in the separation-based scheme, then it is as good as the optimal scheme. More formally, the following corollary follows directly from Theorem \ref{theorem:maintheorem}.

\begin{corollary}\label{corollary:secondcorollary}
If $(D_1,D_2,\ldots,D_K)\in \mathcal{D}(P,b)$, then $(KD_1,KD_2,\ldots,KD_K)^{+}_K\in \hat{\mathcal{D}}(P,b)$, where $(x_1,x_2,...,x_K)^{+}_K=(\min(1,x_1),\min(1,x_2),...,\min(1,x_K))$.
\end{corollary}

Theorem \ref{theorem:maintheorem}, Proposition \ref{proposition:firstcorollary} and the corollaries provide approximate characterizations of the achievable distortion region, essentially stating that the loss of the source-channel separation approach is bounded by constants. The bound on the gap is chosen to be (largely) independent of a specific distortion  tuple on the boundary of $\mathcal{D}(P,b)$, but it will become clear in the next section that such a choice is not necessary.

The proofs of Theorem \ref{theorem:maintheorem} and Proposition \ref{proposition:firstcorollary} rely heavily on the following outer bound, which is one of the main contributions of this work.

\begin{theorem}\label{theorem:outerbound}
Let $\tau_1\geq \tau_2\geq \ldots\geq \tau_{K-1}$ be any $K-1$ non-negative real values, and $\tau_K=0$. If $(D_1,D_2,\ldots,D_K)\in \mathcal{D}(P,b)$, then 
\begin{align}
\sum_{k=1}^K \Delta N_k \left[\frac{(1+\tau_k)\prod_{j=2}^k(D_j+\tau_{j-1})}{\prod_{j=1}^k (D_j+\tau_j)}\right]^{\frac{1}{b}}\leq P+N_1.\label{eqn:outerbound}
\end{align}
\end{theorem}

With the above theorem in mind, let us denote the set of distortion vectors satisfying (\ref{eqn:outerbound}) for a specific choice of $\tau_1\geq \tau_2\geq \ldots\geq \tau_{K-1}$ as $\underline{\mathcal{D}}(P,b,\tau_1,\ldots,\tau_{K-1})$, \textit{i.e.}, (\ref{eqn:defineoutbounds}) as given on the top of next page.
\begin{figure*}[bt]
\normalsize
\setcounter{MYtempeqncnt}{\value{equation}}
\addtocounter{MYtempeqncnt}{1}
\begin{align}
&\underline{\mathcal{D}}(P,b,\tau_1,\ldots,\tau_{K-1})\triangleq\left\{(D_1,D_2,\ldots,D_K):\sum_{k=1}^K \Delta N_k \left[\frac{(1+\tau_k)\prod_{j=2}^k(D_j+\tau_{j-1})}{\prod_{j=1}^k (D_j+\tau_j)}\right]^{\frac{1}{b}}\leq P+N_1,\right.\nonumber\\
&\qquad\qquad\qquad\qquad\qquad\qquad\left.\phantom{\sum_{k=1}^K \Delta N_k \left[\frac{(1+\tau_k)\prod_{j=2}^k(D_j+\tau_{j-1})}{\prod_{j=1}^k (D_j+\tau_j)}\right]^{\frac{1}{b}}} 1\geq D_1\geq D_2 \geq\ldots\geq D_K\geq 0\right\}.\label{eqn:defineoutbounds}
\end{align}
\hrulefill
\end{figure*}
\setcounter{equation}{\value{MYtempeqncnt}}
Thus Theorem \ref{theorem:outerbound} essentially states that $\mathcal{D}(P,b)\subseteq \underline{\mathcal{D}}(P,b,\tau_1,\ldots,\tau_{K-1})$ for any valid choice of $\tau_1,\tau_2,\ldots,\tau_{K-1}$. The following corollary is then immediate.
\begin{corollary}
\label{corollary:outerbound}
\begin{align}
\mathcal{D}(P,b)\subseteq \bigcap_{\tau_{1}\geq \tau_{2}\geq \ldots\geq \tau_{K-1}\geq 0}\underline{\mathcal{D}}(P,b,\tau_1,\ldots,\tau_{K-1}).
\end{align}
\end{corollary}

To illustrate Corollary \ref{corollary:outerbound}, let us consider the case $K=2$ for which the bound involves only one parameter $\tau_1$. For this case, it can be shown through some algebra that this outer bound is equivalent to the one given in \cite{ReznicFederZamir:06}. In Fig. \ref{fig:outerbound}, we illustrate the outer bounds for several specific choices of $\tau_1=\tau$. For comparison, the achievable region using the proposed scheme in \cite{ReznicFederZamir:06} is also given. Note that although the inner bound given by this scheme is extremely close to the outer bound, it appears they do not match exactly. 

\begin{figure*}[tb]
\begin{centering}
\includegraphics[width=17cm]{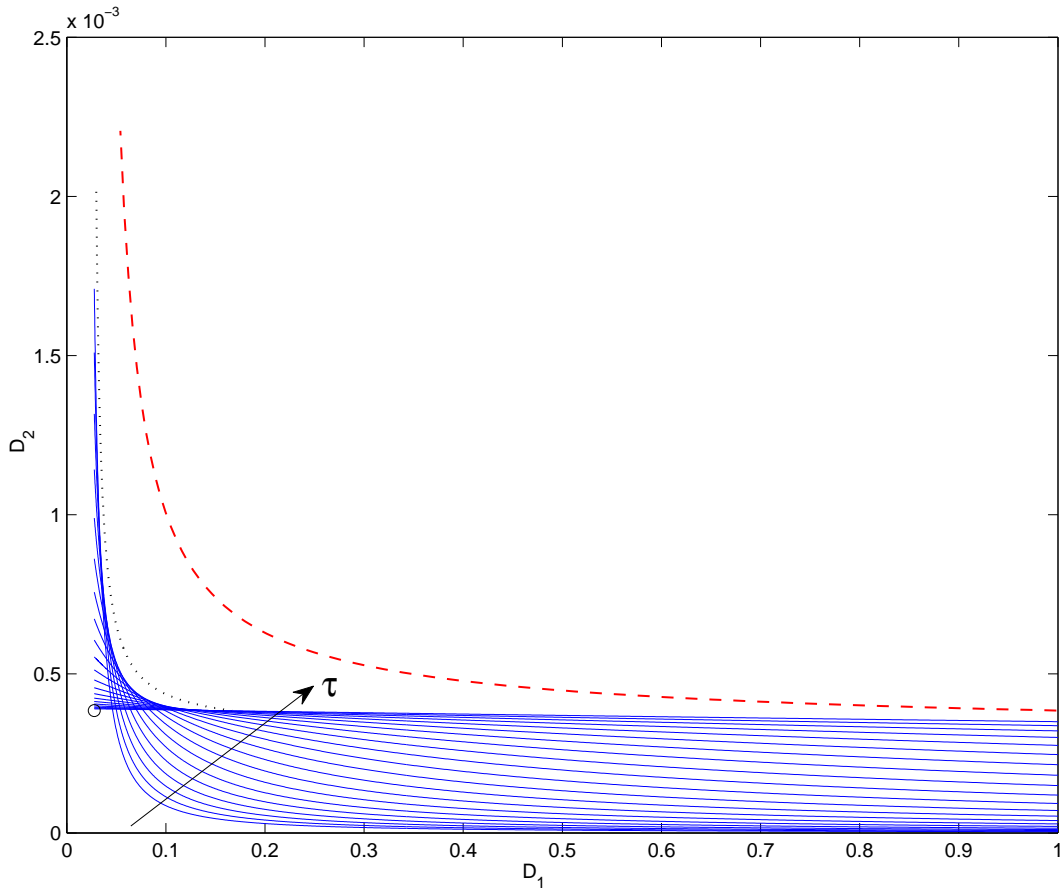}
\caption{Illustration of the outer bounds as $\tau_1=\tau$ varies, for a channel where $N_1= 10$, 
$N_2=1$, $P = 50$, and $b = 2$. The solid blue lines are the outer bounds, and the dashed read line is the inner bound based on source-channel separation; the black circle gives the trivial outer bound with both users at the respective optimum in the  point-to-point setting.\label{fig:outerbound} For comparison, the achievable region using the proposed scheme in \cite{ReznicFederZamir:06} is also given as the dotted black line. }
\end{centering}
\end{figure*}

It is worth emphasizing that we view this outer bound differently from the authors in \cite{ReznicFederZamir:06}: for each possible value of $\tau_1$ we view the condition (\ref{eqn:outerbound}) as specifying an outer bound for the distortion region $(D_1,D_2)$; in contrast, the authors of \cite{ReznicFederZamir:06} viewed the distortion $D_2$ as being lower bounded by a function of $D_1$, and the parameter $\tau_1$ was viewed as an additional variable that is subject to optimization, and consequently only the optimal choice of $\tau_1$ value was of interest. These two views are complementary, however the former view appears to be more natural for the $K$-user problem, which also readily leads to the approximate characterizations. In certain cases, the second view may be more convenient, such as when we are given a specific achievable distortion tuple, and wish to determine how much further improvement is possible or impossible. 

For $K=2$, the properties of the outer bound were thoroughly investigated in \cite{ReznicFederZamir:06}. In certain regimes, this outer bound in fact degenerates for the case of bandwidth compression, and it is looser than the trivial outer bound with each user being optimal in the point-to-point setting\footnote{We would like to thank Dr. Zvi Reznic for clarifying this point in a private communication.}. Due to its non-linear form, the optimization of this bound is rather difficult, and it also appears difficult to determine whether it is always looser than the trivial outer bound in all distortion regimes with bandwidth compression. Nevertheless, it is clear that this outer bound always holds whether the bandwidth is expanded or compressed, and the approximate characterizations are valid in either case. 

A different and simpler approximate characterization may in fact be more useful for the bandwidth compression case\footnote{We would again like to thank David Tse as well as one anonymous reviewer for pointing out this different approximate characterization.}. Consider a different genie who helps the separation-based scheme by giving each individual user half a bit information \textit{per channel use}, and at the same time all the better users also receive this half a bit information for free, then the genie-aided separation-based scheme is as good as the optimal scheme, and moreover each user can in fact achieve the optimal point-to-point distortion.  To see this approximation holds, first observe that the following broadcast channel rates are achievable by using the Gaussian broadcast channel capacity region characterization \cite{Bergmans:74} (it is particularly easy by using the alternative Gaussian broadcast channel capacity characterization given in (\ref{eqn:capacity}))
\begin{align}
&R_k=\max\left(\frac{1}{2}\log_2(1+\frac{P}{N_k})-\frac{1}{2}\log_2(1+\frac{P}{N_{k-1}})-\frac{1}{2},0\right),\nonumber\\
&\qquad\qquad\qquad\qquad\qquad\qquad k=1,2,\ldots,K.\label{eqn:channelapproximation}
\end{align}
The $k$-th user can thus utilize a total rate of $\sum_{i=1}^kR_i$ per channel use on this broadcast channel; together with the genie-provided rates, it will have at least a total rate of $\frac{1}{2}\log_2(1+\frac{P}{N_k})$ per channel use, {\em i.e.}, the optimal point to point channel rate. Since the Gaussian source is successively refinable \cite{EquitzCover:91}, it is now clear that each user can achieve the optimal point-to-point distortion with this genie-aided separation-based scheme. Note that though this approximation is good for bandwidth compression, it can be rather loose when the bandwidth expansion factor is large. In contrast, the approximations given in Theorem \ref{theorem:maintheorem} and Proposition \ref{proposition:firstcorollary} are independent of the bandwidth mismatch factor (the genie provides information in terms of \textit{per source sample}); another difference is that the approximations given in Theorem \ref{theorem:maintheorem} and Proposition \ref{proposition:firstcorollary} rely on the new outer bound, instead of the simple point-to-point distortion outer bound.

It is clear from the above discussion that the outer bound in Theorem \ref{theorem:maintheorem} may be further improved by taking its intersection with the trivial point-to-point outer bound. In the remainder of this paper, we do not pursue such possible improvements, but instead focus on the proofs for the results stated in Theorem \ref{theorem:maintheorem} and Proposition \ref{proposition:firstcorollary}.

\section{Proof of the Main Results for Gaussian Broadcast Channels}
\label{sec:proof}

The proofs of the main results for Gaussian source broadcast on Gaussian broadcast channels are given in this section. We start by establishing a simple inner bound for the distortion region $\mathcal{D}(P,b)$ based on source-channel separation, and then focus on deriving an outer bound, or more precisely a set of outer bounds. The approximate characterizations are then rather straightforward by combining these two bounds. From here on, we shall use natural logarithm for concreteness, though choosing logarithm of a different base does not make any essential difference.

\subsection{A Simple Inner Bound}

The source-channel separation based coding scheme we consider is extremely simple, which is the combination of a Gaussian successive refinement source code and a Gaussian broadcast channel code; this scheme was thoroughly investigated in \cite{TianShamai:07}, and a solution for the optimal power allocation was given to minimize the expected end-user distortion. Since Gaussian broadcast channel is degraded, a better user can always decode completely the messages sent to the worse users, and thus a successive refinement source code is a perfect match for this channel. Note that such a source-channel separation approach is not optimal in general for this joint source-channel coding problem; see for example \cite{Gastpar:03}.

The Gaussian broadcast channel capacity region is well known \cite{Bergmans:74}, which is usually given in a parametric form in terms of the power allocation. In this work, we will use an alternative representation, which first appeared in \cite{Tse:97} and was instrumental for deriving the optimal power allocation solution in \cite{TianShamai:07}. The Gaussian broadcast channel capacity region (per channel use) can be written in the form in (\ref{eqn:capacity}) as given on the top of next page.
\begin{figure*}[bt]
\normalsize
\setcounter{MYtempeqncnt}{\value{equation}}
\addtocounter{MYtempeqncnt}{1}
\begin{align}
&\mathcal{C}=\left\{(R_1,R_2,\ldots,R_M):R_k\geq 0, \,k=1,2,\ldots,K,\quad
\sum_{k=1}^K \Delta N_k\exp\left(2\sum_{j=1}^k R_j\right)\leq P+N_1\right\}.\label{eqn:capacity}
\end{align}
\hrulefill
\end{figure*}
\setcounter{equation}{\value{MYtempeqncnt}}

The rate $R_k$ is the individual message rate intended only to the $k$-th user, however due to the degradedness, all the better users can also decode this message. Since the Gaussian source is successively refinable \cite{EquitzCover:91}, by combining an optimal Gaussian successive refinement source code with a Gaussian broadcast code that (asymptotically) achieves (\ref{eqn:capacity}), we have the following theorem.
\begin{theorem}
\label{theorem:innerbound}
\begin{align}
\hat{\mathcal{D}}(P,b)\subseteq \mathcal{D}(P,b).
\end{align}
\end{theorem}
\begin{proof}
We wish to show that any $(D_1,D_2,\ldots,D_K)\in \hat{\mathcal{D}}(P,b)$ is indeed achievable. Using the separation scheme, we only need to show the channel rates $(R_1,R_2,\ldots,R_K)$ specified by
\begin{align}
D_k=\exp(-2b\sum_{j=1}^k R_j), \quad k=1,2,\ldots,K,
\end{align}
are achievable on this Gaussian broadcast channel. The non-negative vector $(R_1,R_2,\ldots,R_K)$ is uniquely determined by $(D_1,D_2,\ldots,D_K)$, and it is straightforwardly seen that it indeed satisfies the inequality in (\ref{eqn:capacity}). The proof is thus complete. 
\end{proof}

\subsection{An Outer Bound}

Next we derive a set of conditions that any achievable distortion vector $(D_1,D_2,\ldots,D_K)$ has to satisfy, {\em i.e.}, Theorem \ref{theorem:outerbound}.

\begin{proof}[Proof of Theorem \ref{theorem:outerbound}]
Let us first introduce a set of auxiliary random variables, defined as
\begin{align}
U_k=S+V_k,\quad k=1,2,\ldots,K-1,
\end{align}
where $V_k$'s are zero Gaussian random variables, with variance $\tau_k$, and furthermore
\begin{align}
V_k=V_{k+1}+\Delta V_k,\quad k=1,2,\ldots,K-1,
\end{align}
where $\Delta V_k$ is a zero-mean Gaussian random variable, independent of everything else, with variance $\Delta\tau_k=\tau_k-\tau_{k+1}$. For convenience, we define $U_K=S$, which implies $\tau_{K}\triangleq 0$; furthermore, define $U_0\triangleq 0$, \textit{i.e.}, being a constant. This technique of introducing auxiliary random variables beyond those in the original problem was previously used in \cite{Ozarow:80,ReznicFederZamir:06,Tian:08} to derive outer bounds, and specifically in \cite{Tian:08} more than one random variable was introduced, whereas in \cite{Ozarow:80,ReznicFederZamir:06} only one was introduced. 

For any encoding and decoding functions, we consider a quantity which bears some similarity to the expression for the Gaussian broadcast channel capacity (\ref{eqn:capacity}), and we denote this quantity as $E_{f,g}(\tau_1,\tau_2,\ldots,\tau_{K-1})$ due to its sum exponential form
\begin{align}
&E_{f,g}(\tau_1,\tau_2,\ldots,\tau_{K-1})\triangleq \nonumber\\
&\qquad\sum_{k=1}^K \Delta N_k \exp\left[\frac{2}{n}\sum_{j=1}^k I(U^m_j;Y^n_j|U^m_1,U^m_2,\ldots,U^m_{j-1})\right].
\end{align}
The subscript $(f,g)$ makes it clear that this quantity depends on the specific encoding and decoding functions. Next we shall derive universal upper and lower bounds for this quantity regardless the specific choice of functions $(f,g)$, which eventually yield an outer bound for $\mathcal{D}(P,b)$. 

Let $(f,g)$ be any encoding and decoding functions that (asymptotically) achieve the distortions $(D_1,D_2,\ldots,D_K)$. 
We first derive a lower bound for $E_{f,g}(\tau_1,\tau_2,\ldots,\tau_{K-1})$. Observe that for $j=2,3,\ldots,K$,
\begin{align}
&I(U^m_j;Y^n_j|U^m_1,U^m_2,\ldots,U^m_{j-1})\nonumber\\
&\qquad=I(U^m_j;Y^n_j)-I(U^m_{j-1};Y^n_j)\nonumber\\
&\qquad\qquad\geq \frac{m}{2}\log\frac{(1+\tau_j)(D_j+\tau_{j-1})}{(1+\tau_{j-1})(D_j+\tau_{j})},\label{eqn:difference}
\end{align}
where the equality is due to the Markov string $U^m_{1}\leftrightarrow U^m_{2}\leftrightarrow\ldots\leftrightarrow U^m_{K-1}\leftrightarrow S^m\leftrightarrow X^n\leftrightarrow Y^n_j$, and the inequality is by Lemma \ref{lemma:difference}. Moreover, also by Lemma \ref{lemma:difference}, we have 
\begin{align}
I(U^m_1;Y^n_1)\geq  \frac{m}{2}\log\frac{1+\tau_1}{D_1+\tau_1}. \label{eqn:firstone}
\end{align}
It follows that
\begin{align}
&\sum_{j=1}^k I(U^m_j;Y^n_j|U^m_1,U^m_2,\ldots,U^m_{j-1})\nonumber\\
&\geq\frac{m}{2}\log\frac{1+\tau_1}{D_1+\tau_1}+\frac{m}{2}\sum_{j=2}^k \log\frac{(1+\tau_j)(D_j+\tau_{j-1})}{(1+\tau_{j-1})(D_j+\tau_{j})}\nonumber\\
&=\frac{m}{2}\log\frac{1+\tau_k}{D_1+\tau_1}+\frac{m}{2}\sum_{j=2}^k \log\frac{D_j+\tau_{j-1}}{D_j+\tau_{j}}
\end{align}
Summarizing the above bounds, we have 
\begin{align}
\label{eqn:lowerboundE}
&E_{f,g}(\tau_1,\tau_2,\ldots,\tau_{K-1})\nonumber\\
&\,\geq \sum_{k=1}^K \Delta N_k \exp\left[ \frac{1}{b}\log\frac{1+\tau_k}{D_1+\tau_1}+\frac{1}{b}\sum_{j=2}^k \log\frac{D_j+\tau_{j-1}}{D_j+\tau_{j}}\right].
\end{align}

Next we turn to upper-bounding $E_{f,g}(\tau_1,\tau_2,\ldots,\tau_{K-1})$, and first write the following.
\begin{align}
&\frac{2}{n}\sum_{j=1}^k I(U^m_j;Y^n_j|U^m_1,U^m_2,\ldots,U^m_{j-1})\nonumber\\
&=\frac{2}{n}\sum_{j=1}^k \left[I(U^m_j;Y^n_j)-I(U^m_{j-1};Y^n_j)\right]\nonumber\\
&=\frac{2}{n}\sum_{j=1}^k \left[h(Y^n_j|U^m_{j-1})-h(Y^n_j|U^m_j)\right]\nonumber\\
&=\frac{2}{n}\sum_{j=1}^k h(Y^n_j|U^m_{j-1})-\frac{2}{n}\sum_{j=1}^k h(Y^n_j|U^m_j).
\end{align}
Applying the entropy power inequality \cite{CoverThomas} for $j=1,2,\ldots,K-1$, we have 
\begin{align}
&\exp\left[\frac{2}{n}h(Y^n_j|U^m_j)\right]\nonumber\\
&\geq \exp\left[\frac{2}{n}h(Y^n_{j+1}|U^m_j)\right]+\exp\left[\log(2\pi e\Delta N_j)\right]\nonumber\\
&=\exp\left[\frac{2}{n}h(Y^n_{j+1}|U^m_j)\right]+2 \pi e\Delta N_j. \label{eqn:applyentropypower}
\end{align}
For $j=K$, it is clear that
\begin{align}
&\exp\left[\frac{2}{n}h(Y^n_K|U^m_K)\right]=\exp\left[\frac{2}{n}h(Y^n_K|S^m)\right]\nonumber\\
&\qquad\qquad\qquad\qquad\qquad= 2 \pi eN_K=2 \pi e \Delta N_K.
\end{align}
By defining $\exp\left[\frac{2}{n}h(X^n|S^m)\right]\triangleq 0$, it now follows that 
\begin{align}
&E_{f,g}(\tau_1,\tau_2,\ldots,\tau_{K-1})\nonumber\\
&=\sum_{k=1}^K \Delta N_k \exp\left[\frac{2}{n}\sum_{j=1}^k I(U^m_j;Y^n_j|U^m_1,U^m_2,\ldots,U^m_{j-1})\right]\nonumber\\
&\leq \sum_{k=1}^K \Delta N_k \frac{\exp\left[\frac{2}{n}\sum_{j=1}^k h(Y^n_j|U^m_{j-1})\right]}{\prod_{j=1}^k\left[\exp(\frac{2}{n}h(Y^n_{j+1}|U^m_j))+2 \pi e\Delta N_j\right]}.
\end{align}
We bound this summation, by considering the summands in the reversed order, \textit{i.e.}, $k=K,K-1,\ldots,1$.
Starting with the summands when $k=K-1$ and $k=K$, we have (\ref{eqn:firststep}) as given on the top of next page
\begin{figure*}[bt]
\normalsize
\setcounter{MYtempeqncnt}{\value{equation}}
\addtocounter{MYtempeqncnt}{1}
\begin{align}
&\Delta N_{K-1} \frac{\exp\left[\frac{2}{n}\sum_{j=1}^{K-1} h(Y^n_j|U^m_{j-1})\right]}{\prod_{j=1}^{K-1}\left[\exp(\frac{2}{n}h(Y^n_{j+1}|U^m_j))+2 \pi e\Delta N_j\right]}+\Delta N_{K} \frac{\exp\left[\frac{2}{n}\sum_{j=1}^{K} h(Y^n_j|U^m_{j-1})\right]}{\prod_{j=1}^{K}\left[\exp(\frac{2}{n}h(Y^n_{j+1}|U^m_j))+2 \pi e\Delta N_j\right]}\nonumber\\
&=\frac{\exp\left[\frac{2}{n}\sum_{j=1}^{K-1} h(Y^n_j|U^m_{j-1})\right]}{\prod_{j=1}^{K-1}\left[\exp(\frac{2}{n}h(Y^n_{j+1}|U^m_j))+2 \pi e\Delta N_j\right]}\left[\Delta N_{K-1}+\Delta N_{K} \frac{\exp\left[\frac{2}{n}h(Y^n_K|U^m_{K-1})\right]}{2 \pi e\Delta N_K}\right]\nonumber\\
&=\frac{1}{2 \pi e}\frac{\exp\left[\frac{2}{n}\sum_{j=1}^{K-1} h(Y^n_j|U^m_{j-1})\right]}{\Pi_{j=1}^{K-2}\left[\exp(\frac{2}{n}h(Y^n_{j+1}|U^m_j))+2 \pi e\Delta N_j\right]}.\label{eqn:firststep}
\end{align}
\hrulefill
\end{figure*}
\setcounter{equation}{\value{MYtempeqncnt}}
Continuing this line of reduction, we finally arrive at (\ref{eqm:upperboundE}) when $k=1$ 
\begin{figure*}[bt]
\normalsize
\setcounter{MYtempeqncnt}{\value{equation}}
\addtocounter{MYtempeqncnt}{1}
\begin{align}
E_{f,g}(\tau_1,\tau_2,\ldots,\tau_{K-1})\leq 
&\Delta N_1 \frac{\exp\left[\frac{2}{n} h(Y^n_1)\right]}{\exp\left[\frac{2}{n}h(Y^n_{2}|U^m_1))\right]+2 \pi e\Delta N_1}+\frac{1}{2 \pi e}\frac{\exp\left[\frac{2}{n}\sum_{j=1}^{2} h(Y^n_j|U^m_{j-1})\right]}{\left[\exp(\frac{2}{n}h(Y^n_{2}|U^m_1))+2 \pi e\Delta N_1\right]}\nonumber\\
&=\frac{\exp\left[\frac{2}{n} h(Y^n_1)\right]}{\exp\left[\frac{2}{n}h(Y^n_{2}|U^m_1))\right]+2 \pi e\Delta N_1}\left[\Delta N_1+\frac{\exp \left[\frac{2}{n}h(Y^n_2|U^m_1)\right]}{2\pi e}\right]\nonumber\\
&=\frac{\exp\left[\frac{2}{n} h(Y^n_1)\right]}{2\pi e}\leq P+N_1,\label{eqm:upperboundE}
\end{align}
\hrulefill
\end{figure*}
\setcounter{equation}{\value{MYtempeqncnt}}
where the last inequality is by the concavity of the $\log(\cdot)$ function and the given power constraint.

Combining (\ref{eqn:lowerboundE}) and (\ref{eqm:upperboundE}), it is clear that for any encoding and decoding functions $(f,g)$
\begin{align}
&P+N_1\geq E_{f,g}(\tau_1,\tau_2,\ldots,\tau_{K-1})\nonumber\\
&\geq \sum_{k=1}^K \Delta N_k \exp\left[ \frac{1}{b}\log\frac{1+\tau_k}{D_1+\tau_1}+\frac{1}{b}\sum_{j=2}^k \log\frac{D_j+\tau_{j-1}}{D_j+\tau_{j}}\right],
\end{align}
which completes the proof.
\end{proof}

The meaning of the newly introduced random variable $U_k$ can be roughly understood as the message meant for the $k$-th user. Under this interpretation, the term $\frac{1}{n}I(U^m_k;Y^n_k|U^m_1,U^m_2,\ldots,U^m_{k-1})$ in the quantity $E_{f,g}(\tau_1,\tau_2,\ldots,\tau_{K-1})$ essentially represents the individual rate intended for the $k$-th user in the Gaussian broadcast channel; this informal understanding provides the rationale for bounding $E_{f,g}(\tau_1,\tau_2,\ldots,\tau_{K-1})$. This interpretation is nevertheless not completely accurate, and thus the outer bound is likely to be not tight in general, but suffices to provide approximate characterizations.

\subsection{The Approximate Characterizations}

Now we are ready to prove Theorem \ref{theorem:maintheorem} and Proposition \ref{proposition:firstcorollary}. 
\begin{proof}[Proof of Theorem \ref{theorem:maintheorem} and Proposition \ref{proposition:firstcorollary}]
The first inclusion $\hat{\mathcal{D}}(P,b)\subseteq \mathcal{D}(P,b)$ in Theorem \ref{theorem:maintheorem} is simply Theorem \ref{theorem:innerbound}, and thus we focus on the other inclusion $\mathcal{D}(P,b)\subseteq \underline{\mathcal{D}}^*(P,b)\cap \underline{\mathcal{D}}(P,b)$, for which we prove $\mathcal{D}(P,b)\subseteq \underline{\mathcal{D}}^*(P,b)$ and $\mathcal{D}(P,b)\subseteq \underline{\mathcal{D}}(P,b)$ separately. From Theorem \ref{theorem:outerbound}, it is clear that if $(D_1,D_2,\ldots,D_K)\in \mathcal{D}(P,b)$, then (\ref{eqn:outerbound}) holds with any $\tau_1\geq \tau_2\geq \ldots\geq \tau_{K-1}\geq 0$, and thus (\ref{eqn:outerbound}) holds when we choose $\tau_k=D_k$ for $k=1,2,\ldots,K-1$. It follows that the following condition has to be satisfied by any achievable distortion vector
\begin{align}
\sum_{k=1}^K \Delta N_k \left[\frac{(1+D_k)\prod_{j=2}^k(D_j+D_{j-1})}{\prod_{j=1}^k (D_j+D_j)}\right]^{\frac{1}{b}}\leq P+N_1.
\end{align}
However, notice that
\begin{align}
&\sum_{k=1}^K \Delta N_k \left[\frac{(1+D_k)\prod_{j=2}^k(D_j+D_{j-1})}{\prod_{j=1}^k (D_j+D_j)}\right]^{\frac{1}{b}}\nonumber\\
&\qquad\geq \sum_{k=1}^K \Delta N_k \left[\frac{\prod_{j=2}^kD_{j-1}}{\prod_{j=1}^k 2D_j}\right]^{\frac{1}{b}}=\sum_{k=1}^K \Delta N_k (2^kD_k)^{-\frac{1}{b}}.
\end{align}
It now follows straightforwardly that any achievable distortion vector has to satisfy 
\begin{align}
\sum_{k=1}^K \Delta N_k (2^kD_k)^{-\frac{1}{b}}\leq P+N_1,
\end{align}
and $\mathcal{D}(P,b)\subseteq \underline{\mathcal{D}}^*(P,b)$ is proved.

To prove $\mathcal{D}(P,b)\subseteq \underline{\mathcal{D}}(P,b)$, note again that any achievable distortion vector has to satisfy Theorem \ref{theorem:outerbound}, and because of the similarity between the forms given in (\ref{eqn:secondouterbound}) and (\ref{eqn:outerbound}), we only need to prove
\begin{align}
\frac{1}{KD_k}\leq \frac{(1+\tau_k)\prod_{j=2}^k(D_j+\tau_{j-1})}{\prod_{j=1}^k (D_j+\tau_j)},\quad k=1,2,\ldots,K, \label{eqn:Kfactor}
\end{align}
for some specific choice of $\tau_1\geq \tau_2\geq \ldots \geq \tau_{K-1}\geq 0$. We first consider the case that $D_1\leq 1/K$; the case that $D_1>1/K$ needs to be treated in a slightly different manner, as we shall discuss later. 
We take an induction approach, for which the following auxiliary quantities are needed
\begin{align}
\alpha_k= \frac{D_{k}(1+\tau_{k})}{D_{k+1}+\tau_{k}},\qquad k=1,2,\ldots,K-1.
\end{align}
We claim that there exist $\tau_1\geq \tau_2\geq \ldots \geq \tau_{K-1}\geq 0$ such that (\ref{eqn:Kfactor}) holds with equality for $k=1,2,\ldots,K-1$, and holds for $k=K$ with equality or inequality; moreover with these $\tau_k$'s we have $\alpha_k\leq (K-k)^{-1}$. First consider the case $k=1$, since the function 
\begin{align}
\Phi_k(\tau)=\frac{1+\tau}{D_k+\tau}
\end{align}
is monotonically decreasing and continuous in the range $[0,\infty]$, as long as 
\begin{align}
\Phi_1(0)=\frac{1}{D_1} \geq \frac{1}{KD_1}\geq 1=\Phi_1(\infty),
\end{align}
there exists a unique solution of $\tau_1$ such that (\ref{eqn:Kfactor}) holds with equality. This is indeed true for $D_1\leq 1/K$, which gives
\begin{align}
\tau_1=\frac{(K-1)D_1}{1-KD_1}.
\end{align}
It follows that 
\begin{align}
\alpha_1=\frac{D_{1}(1+\tau_{1})}{D_{2}+\tau_{1}}=&\frac{D_1(1-D_1)}{D_2(1-KD_1)+(K-1)D_1}\nonumber\\
&\leq \frac{D_1}{(K-1)D_1}=\frac{1}{K-1},
\end{align}
and thus our claim holds for $k=1$. Next suppose the claim is true for $k=k^*$ and we prove it is also true $k=k^*+1$, for which we wish to find $\tau_{k^*+1}$ such that
\begin{align}
&\frac{1}{KD_{k^*+1}}=\frac{(1+\tau_{k*+1})\prod_{j=2}^{k^*+1}(D_j+\tau_{j-1})}{\prod_{j=1}^{k^*+1} (D_j+\tau_j)}\nonumber\\
&=\frac{(1+\tau_{k*})\prod_{j=2}^{k^*}(D_j+\tau_{j-1})}{\prod_{j=1}^{k^*} (D_j+\tau_j)}\frac{(1+\tau_{k*+1})(D_{k^*+1}+\tau_{k^*})}{(1+\tau_{k*})(D_{k^*+1}+\tau_{k^*+1})}\nonumber\\
&=\frac{1}{KD_{k^*}}\frac{(1+\tau_{k*+1})(D_{k^*+1}+\tau_{k^*})}{(1+\tau_{k*})(D_{k^*+1}+\tau_{k^*+1})},\label{eqn:kstarequal}
\end{align}
where the last equality is by the supposition that the claim holds true for $k^*$. Again by the monotonicity and continuity of $\Phi_k(\tau)$, as long as the choice of $(\tau_1,\tau_2,\ldots,\tau_{k^*})$ satisfies 
\begin{align}
&\Phi_{k^*+1}(0)=\frac{1}{D_{k^*+1}}\geq \frac{D_{k^*}(1+\tau_{k^*})}{D_{k^*+1}(D_{k^*+1}+\tau_{k^*})}\nonumber\\
&\qquad\geq \frac{1+\tau_{k^*}}{D_{k^*+1}+\tau_{k^*}}=\Phi_{k^*+1}(\tau_{k^*}),\label{eqn:inductioninequality}
\end{align}
there exists a valid solution $\tau_{k^*+1}$ in $[0,\tau_{k^*}]$ for (\ref{eqn:kstarequal}) to hold. The second inequality in (\ref{eqn:inductioninequality}) is clearly true, and thus we only need to consider the first inequality. However, notice that
\begin{align}
&\frac{D_{k^*}(1+\tau_{k^*})}{D_{k^*+1}(D_{k^*+1}+\tau_{k^*})}=\frac{\alpha_{k^*}}{D_{k^*+1}}\nonumber\\
&\qquad\qquad\leq \frac{1}{(K-k^*)D_{k^*+1}}\leq \frac{1}{D_{k^*+1}}, 
\end{align}
and we thus conclude that there indeed exists a valid solution of $\tau_{k^*+1}$ for (\ref{eqn:kstarequal}), or more precisely 
\begin{align}
\tau_{k^*+1}=\frac{D_{k^*+1}(1-\alpha_{k^*})}{\alpha_{k^*}-D_{k^*+1}}.
\end{align}
To bound $\alpha_{k^*+1}$, we write
\begin{align}
&\alpha_{k^*+1}=\frac{D_{k^*+1}(1+\tau_{k^*+1})}{D_{k^*+2}+\tau_{k^*+1}}\leq\frac{D_{k^*+1}(1+\tau_{k^*+1})}{\tau_{k^*+1}}\nonumber\\
&\qquad=\frac{\alpha_{k^*}(1-D_{k^*+1})}{1-\alpha_{k^*}}\leq \frac{\alpha_{k^*}}{1-\alpha_{k^*}}\nonumber\\
&\qquad\qquad\qquad\leq \frac{(K-k^*)^{-1}}{1-(K-k^*)^{-1}}=\frac{1}{K-k^*-1},
\end{align}
where the last inequality is by the monotonicity of $x/(1-x)$ in $[0,1)$, and the fact $\alpha_{k^*}\leq (K-k^*)^{-1}$. The induction is thus complete. 
It only remains to check that when $k=K$ the inequality (\ref{eqn:Kfactor}) still holds, for which we have
\begin{align}
&\frac{\prod_{j=2}^{K}(D_j+\tau_{j-1})}{D_K\prod_{j=1}^{K-1} (D_j+\tau_j)}\nonumber\\
&=\frac{(1+\tau_{K-1})\prod_{j=2}^{K-1}(D_j+\tau_{j-1})}{\prod_{j=1}^{K-1} (D_j+\tau_j)}\frac{D_{K}+\tau_{K-1}}{(1+\tau_{K-1})D_{K}}\nonumber\\
&=\frac{1}{KD_{K-1}}\frac{D_{K}+\tau_{K-1}}{(1+\tau_{K-1})D_{K}}=\frac{1}{\alpha_{K-1}KD_{K}}\geq\frac{1}{KD_{K}},
\end{align}
where the last inequality is by the fact $\alpha_{K-1}\leq 1$. We have thus proved that $\mathcal{D}(P,b)\subseteq \underline{\mathcal{D}}(P,b)$ for the case $D_1\leq 1/K$. 

Next we briefly discuss the case $D_1\geq D_2\geq \ldots \geq D_r>1/K\geq D_{r+1}\geq \ldots\geq D_K$, and we shall prove that (\ref{eqn:Kfactor}) holds for some $\tau_1\geq \tau_2\geq \ldots \geq \tau_{K-1}\geq 0$. Notice that by choosing $\tau_1,\tau_2,\ldots,\tau_r$ sufficiently large, we can make 
\begin{align}
\frac{(1+\tau_k)\prod_{j=2}^k(D_j+\tau_{j-1})}{\prod_{j=1}^k (D_j+\tau_j)}\geq \frac{1}{KD_k},\quad k=1,2,\ldots,r. \label{eqn:LHSRHS}
\end{align}
because of the strict inequality given in $D_1\geq D_2\geq \ldots \geq D_r>1/K$, and the fact that the left-hand-side of (\ref{eqn:LHSRHS}) goes to $1$ when we send  $\tau_1,\tau_2,\ldots,\tau_r$ to infinity. We have also $\frac{1}{K-r}\geq \frac{1}{K}\geq D_{r+1}\geq \ldots\geq D_K$, and thus there exist $\tau_{r+1}\geq \ldots\geq \tau_{K-1}$ such that
\begin{align}
\frac{1}{(K-r)D_k}\leq \frac{(1+\tau_k)\prod_{j=r+2}^k(D_j+\tau_{j-1})}{\prod_{j=r+1}^k (D_j+\tau_j)}
\end{align}
holds with equality for $k=r+1,r+2,\ldots,K-1$ and with either equality or inequality for $k=K$, by applying the result for the previously discussed case in a system with $K-r$ users. Since we can choose $\tau_1,\tau_2,\ldots,\tau_r$ sufficiently large, it is clear that indeed this set of $\tau_k$'s makes (\ref{eqn:Kfactor}) hold for $k=r+1,r+2,\ldots,K$. This completes the proof for $\mathcal{D}(P,b)\subseteq \underline{\mathcal{D}}(P,b)$.


To prove Proposition \ref{proposition:firstcorollary}, we need to choose different values for $(\tau_1,\tau_2,\ldots,\tau_{K-1})$. Essentially, when the condition $2D_{k}\leq D_{k-1}$ is not satisfied for certain $k$, we will choose to ignore the contribution of $D_{k}$ in the outer bound of Theorem \ref{theorem:outerbound} by taking an appropriate value of $\tau_{k}$. 
For convenience, define the set $\mathcal{B}=\{k:B_k= 1\}$, where $B_k$ is the labeling function given before Proposition \ref{proposition:firstcorollary}; denote the member of $\mathcal{B}$ in an increasing order as $k_1,k_2,\ldots,k_{|\mathcal{B}|}$, where $|\mathcal{B}|$ is the cardinality of the set $\mathcal{B}$. The value of $\tau_k$'s are set by the following recursive formula
\begin{align}
\tau_k&=\left \{\begin{array}{cc}D_{k}& \mbox{if } B_k=1\\ \tau_{k-1}&\mbox{if }B_k=0\end{array} \right.,\quad k=1,2,\ldots,K-1,
\end{align}
where we have defined $\tau_0\triangleq 1$ for convenience.

Note that $\tau_{k+1}\leq \tau_{k}$ for $k=1,2,\ldots,K-2$. It follows from Theorem \ref{theorem:outerbound} that this achievable distortion vector has to satisfy (\ref{eqn:conditions}) on the top of next page,
\begin{figure*}[bt]
\normalsize
\setcounter{MYtempeqncnt}{\value{equation}}
\addtocounter{MYtempeqncnt}{1}
\begin{align}
&\sum_{k=1}^K \Delta N_k \left[\frac{(1+\tau_k)\prod_{j=2}^k(D_j+\tau_{j-1})}{\prod_{j=1}^k (D_j+\tau_j)}\right]^{\frac{1}{b}}\nonumber\\
&\qquad\qquad\qquad=N_1-N_{k_1}+\sum_{i=1}^{|\mathcal{B}|} (N_{k_i}-N_{k_{i+1}}) \left[\frac{(1+D_{k_i})\prod_{j=2}^{i}(D_{k_j}+D_{k_{j-1}})}{\prod_{j=1}^{i} (D_{k_j}+D_{k_j})}\right]^{\frac{1}{b}}\leq P+N_1,\label{eqn:conditions}
\end{align}
\hrulefill
\end{figure*}
\setcounter{equation}{\value{MYtempeqncnt}}
where for convenience we define $k_{|\mathcal{B}|+1}=K+1$; to see the equality holds, note that the terms for $k\notin \mathcal{B}$ are combined with the terms of $k\in \mathcal{B}$, because this choice of $(\tau_1,\tau_2,\ldots,\tau_{K-1})$ cancels out some of terms in the numerator and the denominator. This implies that the distortion vector has to satisfy
\begin{align}
P+N_1&\geq N_1-N_{k_1}+\sum_{i=1}^{|\mathcal{B}|} (N_{k_i}-N_{k_{i+1}}) (2^iD_{k_i})^{-\frac{1}{b}}\nonumber\\
&=N_1-N_{k_1}+\sum_{i=1}^{|\mathcal{B}|} (N_{k_i}-N_{k_{i+1}}) (D^*_{k_i})^{-\frac{1}{b}}\nonumber\\
&=\sum_{k=1}^{K} \Delta N_{k} (D^*_k)^{-\frac{1}{b}},
\end{align}
where the last two equalities are by the definition of $D^*_k$. The proof can now be completed by applying Theorem \ref{theorem:innerbound}.
\end{proof}

\section{Gaussian Sources on General Broadcast Channels}
\label{sec:general}

In this section, we show that the results for sending Gaussian sources on Gaussian broadcast channels can be conveniently extended to general broadcast channels, which was inspired by a recent work by Avestimehr, Caire and Tse \cite{Avestimehr:08}. 

The broadcast channel is now given by an arbitrary conditional distribution $P(Y_1,Y_2,...,Y_K|X)$, in the alphabets $(\mathcal{Y}_1,\mathcal{Y}_2,\ldots,\mathcal{Y}_K,\mathcal{X})$. For brevity, we omit repeating the definition of the codes here. To distinguish from the Gaussian channel case, we denote the achievable distortion region as $\mathcal{D}_g(b)$, with a bandwidth mismatch factor $b$. For the separation-based scheme, we shall consider combining successive refinement source codes with broadcast codes with degraded message set \cite{KornerMarton:77}. Particularly, for an arbitrary permutation $\pi:\{1,2,\ldots,K\}\rightarrow\{1,2,\ldots,K\}$, the degraded message set requirement implies that there are a total of $K$ independent messages $(W_1,W_2,\ldots,W_K)$, such that the user $\pi(k)$ should decode the messages $W_1,W_2,\ldots,W_k$. For an arbitrary permutation $\pi$, let us denote the achievable distortion by the separation-based approach of combining successive refinement source code with the broadcast code with degraded message set as $\hat{\mathcal{D}}^{\pi}_g(b)$, and the overall achievable distortion region using this separation-based approach is given by
\begin{align}
\hat{\mathcal{D}}_g(b)=\bigcup_{\pi}\hat{\mathcal{D}}^{\pi}_g(b).
\end{align}
Clearly the convex closure of the above region is also achievable, however such generality is not required. 

It is worth noting that since the characterization of the broadcast channel capacity region with degraded message set is still an open problem for $K>2$, we do not have a characterization for the region $\hat{\mathcal{D}}_g(b)$. However, if the broadcast channel is degraded, then only one permutation needs to be considered; moreover, if the capacity region of the broadcast channel (with degraded message set or it is a degraded broadcast channel) is known, such as for the Gaussian broadcast channel case, the region $\hat{\mathcal{D}}_g(b)$ can then be completely characterized.

Now we present the counterpart of Corollary \ref{corollary:firstcorollary}, Proposition \ref{proposition:firstcorollary}, Corollary \ref{corollary:infiniteusers} and Corollary \ref{corollary:secondcorollary} for general broadcast channels. 

\begin{corollary}
\label{corollary:zerocorollaryprime}
If $(D_1,D_2,\ldots,D_K)\in \mathcal{D}_g(P,b)$, and if $D_k\geq 2D_{k+1}$ for $k=1,2,\ldots,K-1$, then $(2D_1,2^2D_2,\ldots,2^KD_K)\in \hat{\mathcal{D}}_g(P,b)$.
\end{corollary}

\begin{proposition}
\label{proposition:firstcorollaryprime}
Let $(D^*_1,D^*_2,\ldots,D^*_K)$ be the relaxed distortion vector of $(D_1,D_2,\ldots,D_K)$ as defined in (\ref{eqn:enhanceddistortion}). If $(D_1,D_2,\ldots,D_K)\in \mathcal{D}_g(b)$, then $(D^*_1,D^*_2,\ldots,D^*_K)\in \hat{\mathcal{D}}_g(b)$.
\end{proposition}

\begin{corollary}
For an infinite number of users indexed by $x$ with $\inf_{\{D_x\}\in\mathcal{D}_g(b)}\inf\{D_x\}\geq d_{\mbox{\small{min}}}>0$, let $\{D^*_x\}$ be the relaxed distortion vector of $\{D_x\}$ as defined in (\ref{eqn:enhanceddistortion}). If $\{D_x\}\in \mathcal{D}_g(b)$, then $\{D^*_x\}\in \hat{\mathcal{D}}_g(b)$, and furthermore, $\sup\{\frac{D_x}{D^*_x}\}\leq \frac{4}{d_{\mbox{\tiny{min}}}}$. 
\label{corollary:firstcorollaryprime}
\end{corollary}

\begin{corollary}\label{corollary:secondcorollaryprime}
If $(D_1,D_2,\ldots,D_K)\in \mathcal{D}_g(P,b)$, then $(KD_1,KD_2,\ldots,KD_K)^{+}_K\in \hat{\mathcal{D}}_g(P,b)$.
\end{corollary}

We only prove Corollary \ref{corollary:secondcorollaryprime} here, since the proofs of Corollary \ref{corollary:zerocorollaryprime}, Proposition \ref{proposition:firstcorollaryprime} and Corollary \ref{corollary:firstcorollaryprime} are quite similar.

\begin{proof}
Assume a distortion vector $(D_1,D_2,\ldots,D_K)$ is indeed achievable with certain joint source-channel coding scheme. Let us view the induced random mapping from $\mathcal{S}^m$ to $\mathcal{Y}^n_i$, $i=1,2,...,K$, by this joint source-channel code as a super-broadcast-channel; without loss of generality, let us assume $D_1\geq D_2\geq \ldots\geq D_K>0$. 

We pick up the story from (\ref{eqn:difference}) and (\ref{eqn:firstone}), and claim that there exists a degraded message set broadcast code on the super-broadcast-channel with asymptotic rate per each length-$m$ block as follows
\begin{align}
&R^m_1=\frac{m}{2}\log\frac{1+\tau_1}{D_1+\tau_1},\nonumber\\
&R^m_j=\frac{m}{2}\log\frac{(1+\tau_j)(D_j+\tau_{j-1})}{(1+\tau_{j-1})(D_j+\tau_{j})}, \qquad j=2,3,\ldots,K.
\end{align}

It is not difficult to see that the distribution $U^m_1,U^m_2,\ldots,U^m_K,X^m$ can be used to construct a well-known super-position code with the above rates \cite{CoverThomas} by (\ref{eqn:difference}) and (\ref{eqn:firstone}). This code needs to span over $n'$ blocks, and note that the induced super-broadcast-channel is block-wise memoryless. We only need to confirm that receiver $k$ can decode all the messages up to the $k$-th layer by successive decoding. Observe that since
\begin{align}
I(U^m_1;Y^n_k)\geq \frac{m}{2}\frac{1+\tau_1}{D_k+\tau_1}\geq \frac{m}{2}\frac{1+\tau_1}{D_1+\tau_1}=R^m_1,
\end{align}
indeed receiver $k$ can decoder the first layer code, and thus recover the codewords based on $U_1$. The above inequality essentially shows that the channel to $Y^n_k$ is more powerful than that to $Y^n_1$, with the given channel input distribution, although the broadcast channel itself is not necessarily degraded. For the $j$-th layer where $j\leq k$, we have
\begin{align}
&I(U^m_j;Y^n_k|U^m_{j-1})\geq \frac{m}{2}\log\frac{(1+\tau_j)(D_k+\tau_{j-1})}{(1+\tau_{j-1})(D_k+\tau_{j})}\nonumber\\
&\qquad\qquad\qquad\geq \frac{m}{2}\log\frac{(1+\tau_j)(D_j+\tau_{j-1})}{(1+\tau_{j-1})(D_j+\tau_{j})}=R^m_j,
\end{align}
because of the monotonicity of the function
\begin{align}
f(D)=\frac{D+\tau_{j-1}}{D+\tau_{j}},
\end{align}
when $\tau_{j-1}\geq \tau_j$, and the fact $D_k\leq D_j$. Thus our claim is indeed true. 

Using this set of degraded message set broadcast channel codes, we can achieve (or more precisely, approach arbitrarily close to) the following distortion $D^*_k$ for any $\tau_1\geq \tau_2\geq \ldots\geq \tau_{K-1}\geq 0$, and $\tau_K\triangleq 0$,
\begin{align}
D^*_k=\exp(-2\sum_{j=1}^k\frac{R^m_j}{m})=\frac{\prod_{j=1}^{k} (D_j+\tau_j)}{(1+\tau_{k})\prod_{j=2}^{k}(D_j+\tau_{j-1})}.
\end{align}
However, in the proof of Theorem \ref{theorem:maintheorem}, we have already showed that there exist $\tau$'s such that (\ref{eqn:Kfactor}) holds, {\em i.e.},
\begin{align} 
D^*_k\leq KD_k, \qquad k=1,2,\ldots,K.
\end{align}
The proof is thus complete. 
\end{proof}

\section{Concluding Remarks}
\label{sec:conclusion}

We derived a new outer bound to the achievable distortion region for the joint source-channel coding problem of sending a Gaussian source on Gaussian broadcast channels with bandwidth mismatch. When combined with a simple separation-based achievability result, this new bound leads to approximate characterizations of the achievable distortion region within some universal constant multiplicative factors. These results are further extended to the case of Gaussian source broadcast on general broadcast channels.

The outer bound was not fully optimized, which seems to be a difficult problem by itself. It may be beneficial to investigate the outer bound more thoroughly when more powerful achievability schemes become available. In the current work, we only considered the separation-based scheme that yields approximate characterizations. 

The technique used in this work can be applied to the problem of multi-source broadcast on more complex communication networks under certain conditions.  We believe that similar results hold for many classes of suitably well-behaved networks. In a follow-up work to this paper, we have shown approximate
separation for a class of such networks along with other results on source-channel separation \cite{TianCDS:10}. 

The outer bounding technique of introducing auxiliary random variables used in this work is inspired by those used in \cite{Ozarow:80, ReznicFederZamir:06,Tian:08}. We believe this technique is also promising for other multi-user information theoretic problems. The role of the auxiliary random variable introduced in \cite{Ozarow:80, ReznicFederZamir:06} was not quite well understood or interpreted previously, and may even appear mysterious to many researchers less familiar with the specific problems being treated. The current work (see also \cite{Tian:08}) has made the meaning of the introduced random variables explicit. More specifically, they are introduced to either substitute the messages in the channel coding problem, or to substitute the source reconstructions in the source coding problem. By this substitution, the quantities representing rates are replaced by the corresponding information quantities. With the interpretation made clear in a general manner, it is our hope that this technique can inspire other meaningful results in the future.

\appendix
\label{appendix:lemmadifference}
\begin{proof}[Proof of Lemma \ref{lemma:difference}]
Define $Z'=V+V'$ and $Z=V$, and thus $U'=S+Z'$ and $U=S+Z$. To prove the first statement, we consider the following chain of inequalities
\begin{align*}
&I(W;U'^m)\\
&=mh(U')-h(U'^m|W)\\
&=mh(U')-h(S^m+Z'^m|W)\\
&=mh(U')-h(S^m+Z'^m-\hat{S}^m|W)\\
&\stackrel{(a)}{\geq} mh(U')-h(S^m+Z'^m-\hat{S}^m)\\
&\stackrel{(b)}{\geq} mh(U')-\sum_{i=1}^m h[S(i)+Z'(i)-\hat{S}(i)]\\
&\stackrel{(c)}{\geq} mh(U')\\
&\qquad-\sum_{i=1}^m\frac{1}{2}\log\left\{(2\pi e)\Expt[(S(i)+Z'(i)-\hat{S}(i))^2] \right\}\\
&= mh(U')-\sum_{i=1}^m\frac{1}{2}\log\left[(2\pi e)(\Expt d(S(i),\hat{S}(i))+\tau') \right],
\end{align*}
where $\hat{S}^m$ is the reconstruction with $W$, and its $i$-th position is denoted as $\hat{S}(i)$. The inequality (a) is because conditioning reduces entropy, (b) is because of the chain rule for differential entropy and the fact that conditioning reduces entropy, and (c) is because Gaussian distribution maximizes the differential entropy for a given second moment. 
Since $\log(\cdot)$ is a concave function, we have 
\begin{align*}
&\sum_{i=1}^m\frac{1}{2}\log\left[(2\pi e)(\Expt d(S(i),\hat{S}(i))+\tau') \right]\nonumber\\
&\qquad\leq \frac{m}{2}\log\left[2\pi e \left(\Expt d(S^n,\hat{S}^n)+\tau'\right)\right].
\end{align*}
It follows
\begin{align*}
I(W;U'^m)&\geq mh(U')-\frac{m}{2}\log\left[2\pi e \left(\Expt d(S^m,\hat{S}^m)+\tau'\right)\right]\\
&\geq mh(U')-\frac{m}{2}\log\left[2\pi e(D+\tau')\right]\\
&=\frac{m}{2}\log\frac{1+\tau'}{D+\tau'},
\end{align*}
which is the first claim in the lemma. 

To prove the second claim, we write the following
\begin{align*}
&I(W;U^m)-I(W;U'^m)\\
&\qquad\qquad=mh(U)-mh(U')+h(U'^m|W)-h(U^m|W).
\end{align*}
For the latter two terms, we have
\begin{align*}
h(U'^m|W)-h(U^m|W)&\stackrel{(a)}{=}h(U'^m|W)-h(U^m|V'^m,W)\\
&\stackrel{(b)}{=}h(U'^m|W)-h(U'^m|V'^m,W)\\
&=I(U'^m;V'^m|W),
\end{align*}
where (a) is because $V'^m$ is independent of $U^m$ and $W$; (b) is by the definition of $U'$. Continuing the chain of inequalities, we have
\begin{align*}
&I(U'^m;V'^m|W)\\
&\stackrel{(a)}{=}h(V'^m)-h(V'^m|S^m+V^m+V'^m,W)\\
&=h(V'^m)-h(V'^m|S^m+V^m+V'^m,\hat{S}^m,W)\\
&\stackrel{(b)}{\geq}h(V'^m)-h(V'^m|S^m-\hat{S}^m+V^m+V'^m)\\
&\stackrel{(c)}{\geq}\sum_{i=1}^m \Big{[}h(V'(i))-h\Big{(}V'(i)|S(i)-\hat{S}(i)+V(i)+V'(i)\Big{)}\Big{]}\\
&=\sum_{i=1}^m I\Big{(}V'(i);S(i)-\hat{S}(i)+V(i)+V'(i)\Big{)}\\
&\stackrel{(d)}{\geq}\sum_{i=1}^m \frac{1}{2}\log \frac{\Expt d(S(i),\hat{S}(i))+\tau'}{\Expt d(S(i),\hat{S}(i))+\tau}\\
&\stackrel{(e)}{\geq} \frac{m}{2}\log \frac{D+\tau'}{D+\tau},
\end{align*}
where (a) is because $V'$ is independent of $W$; (b) is because conditioning reduces entropy; (c) is by applying the chain rule, and the facts that $V'^m$ is an i.i.d. sequence and conditioning reduces entropy; (d) is by applying the mutual information game result that Gaussian noise is the worst additive noise under a variance constraint \cite{CoverThomas} (pg. 263, ex. 1), and taking $V'(i)$ as channel input; finally (e) is due to the convexity and monotonicity of $\log\frac{x+\tau'}{x+\tau}$ in $x\in (0,\infty)$ when $\tau'\geq \tau\geq 0$. This completes the proof for the second claim.
\end{proof}

\section*{Acknowledgments}
The authors would like to thank David Tse for the stimulating discussions at several occasions as well as his insightful comments. The authors are also grateful to the anonymous reviewers for their comments. 

\bibliographystyle{IEEEbib}

\begin{thebibliography}{11}

\bibitem{Shannon:48}
C. E. Shannon,
\newblock ``A mathematical theory of communication,''
\newblock {\em Bell System Technical Journal}, vol. 27, pp. 379--423, pp. 623--656, Jul., Oct. 1948.

\bibitem{Goblick:65}
T. J. Goblick,
\newblock ``Theoretical limitations on the transmission of data from analog sources,''
\newblock {\em IEEE Trans. Information Theory}, vol. 11, no. 4, pp. 558--567, Oct. 1965.

\bibitem{Cover:80}
T. Cover, A. E. Gamal, and M. Salehi,
\newblock ``Multiple access channels with arbitrarily correlated sources,''
\newblock {\em IEEE Trans. Information Theory}, vol. 26, no. 6, pp. 648--657, Nov. 1980.


\bibitem{Gastpar:03}
M. Gastpar, B. Rimoldi, and M. Vetterli, ``To code, or not to code: lossy source-channel communication revisited,''
\emph{IEEE Trans. Information Theory}, vol.~49, no.~5, pp. 1147--1158, May 2003.  


\bibitem{ShamaiVerduZamir:98}
S. Shamai, S. Verdu and R. Zamir,
\newblock ``Systematic lossy source/channel coding,''
\newblock {\em IEEE Trans. Information Theory}, vol. 44, no. 3, pp. 564--579, Mar. 1998.

\bibitem{Mittal:02}
U.~Mittal and N.~Phamdo, ``Hybrid digital-analog ({HDA}) joint source-channel
  codes for broadcasting and robust communications,'' \emph{IEEE Transactions
  on Information Theory}, vol.~48, no.~5, pp. 1082--1102, May 2002.

\bibitem{Skoglund:06}
M.~Skoglund, N.~Phamdo, and F.~Alajaji, ``Hybrid digital-analog source-channel
  coding for bandwidth compression/expansion,'' \emph{IEEE Transactions on
  Information Theory}, vol.~52, no.~8, pp. 3757--3763, Aug. 2006.

\bibitem{ReznicFederZamir:06}
Z. Reznic, M. Feder and R. Zamir,
\newblock ``Distortion bounds for broadcasting with bandwidth expansion,''
\newblock {\em IEEE Trans. Information Theory}, vol. 52, no. 8, pp. 3778--3788, Aug. 2006.

\bibitem{Prabhakaran:05}
V. M. Prabhakaran, R. Puri and K. Ramachandran,
\newblock ``Hybrid analog-digital strategies for source-channel broadcast,''
\newblock {\em Proc. 43rd Allerton Conference on Communication, Control and Computing}, Allerton, IL, Sep. 2005, pp. 94--113.


\bibitem{NarayananCaireWilson:07}
K. Narayanan, G. Caire and M. Wilson,
\newblock ``Duality between broadcasting with bandwidth expansion and bandwidth compression,''
\newblock {\em Proc. IEEE International Symposium on Information Theory}, Nice, France, Jul. 2007, pp. 1161--1165.

  
\bibitem{Caire:ITA06}
G. Caire,
\newblock ``Distortion region in common source broadcasting,'' 
\newblock Open question session, {\em Information Theory and Applications Inaugural Workshop}, San Diego, CA, Feb. 2006.

\bibitem{Ozarow:80}
L.~Ozarow,
\newblock ``On a source-coding problem with two channels and three receivers,''
\newblock {\em Bell Syst. Tech. Journal}, vol. 59, pp. 1909--1921, Dec. 1980.


\bibitem{Tian:08}
C.~Tian, S.~Mohajer, and S. Diggavi,
\newblock ``Approximating the Gaussian multiple description rate region under symmetric distortion constraints,''
\newblock {\em IEEE Trans. Information Theory}, vol. 55, no. 8, pp.  3869--3891, Aug. 2009.


\bibitem{Avestimehr:08}
S. Avestimehr, G. Caire and D.N.C. Tse
\newblock ``On source-channel separation in networks,''
\newblock arXiv:0901.2082v1.

\bibitem{EquitzCover:91}
W.~H.~R. Equitz and T.~M. Cover, ``Successive refinement of information,''
  \emph{IEEE Trans. Information Theory}, vol.~37, no.~2, pp. 269--275, Mar.
  1991.

\bibitem{TianShamai:07}
C. Tian, A. Steiner, S. Shamai, and S. Diggavi,
\newblock ``Successive refinement via broadcast: optimizing expected distortion of a Gaussian source over a Gaussian fading channel,''
\newblock {\em IEEE Trans. Information Theory}, vol. 54, no. 7, pp. 2903--2918, Jul. 2008.

\bibitem{Bergmans:74}
P.~Bergmans,
\newblock ``A simple converse proof for the broadcast channels with additive white Gaussian noise,''
\newblock {\em IEEE Trans. Information Theory}, vol. 20, no. 2, pp. 279--280, Mar. 1974.

\bibitem{Tse:97}
D.~N.~C. Tse, ``Optimal power allocation over parallel Gaussian broadcast
  channels,'' in \newblock {\em U.C. Berkeley technical report UCB/ERL M99/7, 1999}; available at \href{http://www.eecs.berkeley.edu/Pubs/TechRpts/1999/3578.html.}{http://www.eecs.berkeley.edu/Pubs/TechRpts/1999/3578.html}
  
\bibitem{CoverThomas}
T.~M. Cover and J.~A. Thomas,
\newblock {\em Elements of Information Theory},
\newblock New York: Wiley, 1991.

\bibitem{KornerMarton:77}
J. Korner, K. Marton, ``General broadcast channels with degraded message sets,''
  \emph{IEEE Trans. Information Theory}, vol.~23, no.~1, pp. 60--64, Mar. 1977.  

\bibitem{TianCDS:10}
C. Tian, J. Chen, S. Diggavi and S. Shamai,
\newblock ``Optimality and approximate optimality of source-channel separation in networks,''
\newblock {\em Proc. IEEE International Symposium on Information Theory},  Austin, TX, USA, Jun. 2010, pp. 495--499. \newblock see also \href{http://arxiv.org/abs/1004.2648}{http://arxiv.org/abs/1004.2648}


\end{thebibliography}

\begin{biographynophoto}{Chao Tian}(S'00, M'05) received the B.E. degree in Electronic Engineering from Tsinghua University, Beijing, China, in 2000 and the M.S. and Ph. D. degrees in Electrical and Computer Engineering from Cornell University, Ithaca, NY in 2003 and 2005, respectively. 

Dr. Tian was a postdoctoral researcher at Ecole Polytechnique Federale de Lausanne (EPFL) from 2005 to 2007. He joined AT\&T Labs--Research, Florham Park, New Jersey in 2007, where he is now a Senior Member of Technical Staff. His research interests include multi-user information theory, joint source-channel coding, quantization design and analysis, as well as image/video coding and processing.
\end{biographynophoto}

\begin{biographynophoto}{Suhas N. Diggavi}
(S'93, M'99) received the B. Tech. degree in electrical
engineering from the Indian Institute of Technology, Delhi, India, and
the Ph.D. degree in electrical engineering from Stanford University,
Stanford, CA, in 1998.

After completing his Ph.D., he was a Principal Member Technical Staff 
in the Information Sciences Center, AT\&T Shannon Laboratories, Florham
Park, NJ. Since then he had been in the faculty of the School of
Computer and Communication Sciences, EPFL, where he directed the
Laboratory for Information and Communication Systems (LICOS). He is
currently a Professor, in the Department of Electrical Engineering, at
the University of California, Los Angeles. His research interests
include wireless communications networks, information theory, network
data compression and network algorithms.

He is a recipient of the 2006 IEEE Donald Fink prize paper award, 2005
IEEE Vehicular Technology Conference best paper award and the Okawa
foundation research award.  He is currently an editor for ACM/IEEE
Transactions on Networking and IEEE Transactions on Information
Theory. He has 8 issued patents.

\end{biographynophoto}

\begin{biographynophoto}{Shlomo Shamai (Shitz)}(S'80, M'82, SM'89, F'94) 
received the B.Sc., M.Sc., and Ph.D. degrees in
electrical engineering from the Technion---Israel Institute of Technology,
in 1975, 1981 and 1986 respectively.

During 1975-1985 he was with the Communications Research Labs
in the capacity of a Senior Research Engineer. Since 1986 he is with
the Department of Electrical Engineering, Technion---Israel Institute of
Technology, where he is now the William Fondiller Professor of Telecommunications.
His research interests encompasses a wide spectrum of topics in information
theory and statistical communications.

Dr. Shamai (Shitz) is an IEEE Fellow, and the recipient of the 2011
Claude E. Shannon Award.
He is the recipient of the 1999 van der Pol Gold Medal of the Union Radio
Scientifique Internationale (URSI), and a
co-recipient of the 2000 IEEE Donald G. Fink Prize Paper Award, the 2003, and
the 2004 joint IT/COM societies paper award, the 2007 IEEE Information
Theory Society Paper Award, the 2009 European Commission FP7, Network of
Excellence in Wireless COMmunications (NEWCOM++) Best Paper Award,
and the 2010 Thomson Reuters Award for International Excellence
in Scientific Research. He is
He is also the recipient of
1985 Alon Grant for distinguished young scientists and the 2000 Technion Henry
Taub Prize for Excellence in Research.
He has served as Associate Editor for the {\sc Shannon Theory of the IEEE
Transactions on Information Theory}, and has also served on the Board of
Governors of the Information Theory Society.
\end{biographynophoto}

\end{document}